\documentclass[11 pt]{amsart}

\usepackage{amsmath,amssymb,amscd,amsthm,amsxtra}

\usepackage{graphicx, mathrsfs}
\usepackage{amssymb,amsmath,amsthm}
\usepackage{mathrsfs,dsfont}

\theoremstyle{plain}
\newtheorem{theorem}{Theorem} [section]

\newtheorem{remark}[theorem]{Remark}

\numberwithin{equation}{section} 

\newcommand{\middlefig}{.75\textwidth}

\def\he{\rm {h}}
\def \ve{\rm{v}}

\def\RR{{\mathbb R}}
\def\R{{\mathbb R}}

\newcommand{\p}{\partial}

\theoremstyle{remark}

\begin{document}

\title[Hyperbolic NLS]{Radial standing and self-similar waves for the hyperbolic cubic NLS in 2D}
\date{}

\author[Kevrekidis]{Panayotis Kevrekidis$^1$}
\address{$^1$ Department of Mathematics\\University of Massachusetts\\710 N. Pleasant Street\\Amherst MA 01003}
\email{kevrekid@math.umass.edu}
\thanks{$^1$ The first author is funded in part by NSF-DMS-0349023,
NSF DMS-0806762, NSF-CMMI-1000337 and by the Alexander von Humboldt 
and Alexander S. Onassis Foundations. He also acknowledges numerous
discussions on this theme with Kody J.H. Law.}
\author[Nahmod]{Andrea R. Nahmod$^2$}
\address{$^2$ Radcliffe Institute for Advanced Study, Harvard University\\Byerly Hall, 8 Garden Street\\
Cambridge, MA 02138 and
Department of Mathematics\\
University of Massachusetts\\ 710 N. Pleasant Street\\ Amherst MA 01003}
\email{nahmod@math.umass.edu}
\thanks{$^2$ The second author is funded in part by NSF DMS 0803160 and a 2009-2010 Radcliffe Institute for Advanced Study Fellowship.}
\author[Zeng]{Chongchun Zeng$^3$}
\address{$^3$ School of Mathematics \\
Georgia Institute of Technology\\ Atlanta, GA 30332}
\email{zengch@math.gatech.edu}
\thanks{$^3$ The third author is funded in part by NSF-DMS 0801319.}

\begin{abstract} In this note we propose a new set of coordinates to study 
the hyperbolic or non-elliptic cubic nonlinear Schr\"odinger equation in two 
dimensions. Based on these coordinates, we study the
existence of bounded and continuous hyperbolically radial standing waves, as well as hyperbolically radial
self-similar solutions.
Many of the arguments can easily be adapted to more general nonlinearities.
\end{abstract}
\maketitle

\section{Introduction}

The hyperbolic cubic NLS equation has been studied in \cite{GHS1,GHS2,GHS3}
(see also \cite{SuSu,CoDiTri,CoTri} and references therein) and has received
attention by  both the physics and the applied mathematics communities.
This model has been used for examining
the evolution of optical pulses in normally dispersive (quasi-discrete)
optical waveguide array  structures \cite{DROU,LAH}, as well as more generally
in normally dispersive optical media \cite{trillo1,trillo2}. These
studies have, in turn,  motivated a number of  works examining the
type of coherent structures, such as the experimentally tractable X-waves
\cite{trillo1,trillo2,CoTri}, as well as more elaborate structures such
as dark-bright \cite{HAYA} or vortex-bright solitary waves \cite{EFR}.
On the one hand, the spontaneous emergence of nonlinear structures
such as the X-waves has been revealed by means of mechanisms such
as the modulational instability \cite{co1} and their nonlinear
dynamics has been considered \cite{co2}, but also more recently
studies have begun to systematically address their spatio-temporal
interactions \cite{moshonas}.

The equation arises when to the normally-dispersive (1+1)-NLS,
effects due to the diffraction
(in the quasi-discrete setting) or anomalous dispersion 
(in the continuum setting) are incorporated
as an additional term. The $(2+1)$-dimensional quasi-discrete
equation~\cite{DROU,LAH} is described as
\begin{eqnarray}
i  \frac{\partial U_n}{\partial z} -  \frac{\beta_2}{2} \frac{\partial^2 U_n }{\partial t^2} + C ( U_{n+1} + U_{n-1}) + \gamma |U_n|^2 U_n = 0
\label{eqn1}
\end{eqnarray}
 where $U$ is the complex amplitude that represents an envelope of the
relevant wave field,  $z$ is the  longitudinal direction along which
propagation occurs;  and $n$ and $t$ are the two transverse axes (which
represent space and time respectively). In that same vein, it is relevant
to mention that very recently also $(3+1)$-dimensional generalizations
of the above quasi-discrete setting have emerged, giving rise to
the so-called discrete-continuous X-waves in photonic
lattices~\cite{heinrich}.
%Of special interest is the question of existence of {\it symbions} which are solutions of symbiotic form of the dark and bright solitons.

The continuous version of the $(2+1)$-dimensional equation can be
written as
\begin{equation*}
i u_t + u_{xx} - u_{yy} + \gamma |u|^2 u = 0
\qquad \qquad \mbox{ (3HNLS)}
\end{equation*}
and  is related to both the Davey-Stewartson and
Ishimori systems  \cite{GHS2,SuSu}.   In fact, in \cite{GHS1}
Ghidaglia and Saut showed that (3HNLS) satisfies the same linear Strichartz
estimates as its elliptic counterpart and relying on them,  they
proved the hyperbolic-elliptic Davey-Stewartson is locally well posed
in $L^2(\RR^2)$ with time of existence depending on the profile of the
initial data and globally well posed for sufficiently small data.  The same
argument proves that (3HNLS) enjoys the same well posedness results.   In  \cite{GHS3},
Ghidaglia and Saut showed that there are no nontrivial localized standing wave solutions
to (3HNLS); more precisely they proved that there are no solutions of the
form $e^{i \mu t} \phi(x)$ with $\phi \in H^1\cap H^2_{loc}$ and
any $\mu \in \RR$ (regardless of dimension).
From the physical point of view (3HNLS) or  equation (\ref{3HNLS}) below represents the
canonical model for the study (experimentally and also numerically)
of X-waves, as considered in \cite{trillo1,trillo2}; see also the
earlier and more recent reviews of \cite{CoDiTri,CoTri}. It should
also be mentioned that the linear counterpart of the equation
has also been considered extensively~\cite{CoTri,Efre2}
and has given rise to interesting
recent experimentally examined structures such as the linear alternative
to light bullets proposed in~\cite{dnc0} in the form of
Airy-Bessel wavepackets.

The major differences between the regular Laplacian $\Delta = \p_{xx} + \p_{yy}$ and
the hyperbolic Laplacian $\square_{xy} = \p_{xx} - \p_{yy}$ include that a)
the latter has
nontrivial characteristics; b) the former is
invariant under the Euclidean rotations on $\R^2$ while the latter is invariant
under the hyperbolic
rotations; and c) the former corresponds to the positive Dirichlet energy and
the latter to
the energy $E(u):= \int_{\RR^2}  |u_x|^2 - |u_y|^2 dx dy $ which is indefinite. Due to observation a),
this PDE system may have singularities forming along the characteristics and thus we first
find some explicit compatibility conditions along the discontinuities of weak solutions which are
piecewise smooth. Observation b)
naturally motivates us to consider the hyperbolic coordinates $(a, \alpha)$, given by
$x=\pm a \cosh \alpha$ and $y=\pm a \sinh \alpha$. Since the hyperbolic NLS is invariant
under the hyperbolic rotations $(a, \alpha) \to (a, \alpha + \alpha_0)$, it is natural
to start our study with solutions with the corresponding symmetry -- continuous and piecewise
smooth hyperbolically radial (or equivariant) standing waves and hyperbolically radial (or equivariant)
self-similar waves\footnote{By self-similar waves we mean standing waves in the self-similar variables; c.f. \eqref{self-similar} \eqref{SSstanding}.} -- as our initial attempts to further our understanding of
this PDE.
%The existence of these waves is contrasted to the nonexistence of standing waves
%proved in \cite{GHS3} and described above.
They also usually play important roles in studying the dynamics of (3HNLS). In particular, states
which are hyperbolically radial share some similarity with the profiles of the
two-dimensional analogs of the so-called $X$-waves
observed in numerics; as a relevant example see e.g.
Figs. 18.6-18.7 of~\cite{CoDiTri} (cf. also~\cite{CoTri}).
Since the measure $dxdy = |a| da d\alpha$, non-zero
hyperbolic radial states do not belong to $L^2 (\R^2)$. We do not view this as an extremely undesirable
situation. On the one hand, some important elementary waves, like plane waves
and kinks, do have infinite $L^2$ norm.
%Here radial or equivariant waves are some kind of nonlinear version of
%plane waves of the (3HNLS).
On the other hand, since the hyperbolic Laplacian $\square_{xy}$ does not correspond to a positive energy
and has characteristics, it is reasonable, at least as an initial step, to relax the finite energy
requirement and look for hyperbolically radial standing waves which may have asymptotic values other
than zero at spatial infinity in each hyperbolic quadrant.

In constructing these hyperbolically radial (or equivariant) waves, we encounter an ODE
with a regular singular
point. Though for this specific equation there are many standard results concerning the existence of
local solutions, we provide a different proof based on the invariant manifold theory in dynamical systems
which is a generalization of a trick used in \cite{DTZ}. Unlike the power series method or variational
method, this approach requires minimal smoothness and no extra structure.
Our aim in the present paper is to use such techniques and,
especially, the hyperbolic coordinate system in order to construct, for
the first time to our knowledge,
nontrivial standing wave solutions to the (3HNLS) equation.
In addition, we also consider radial self-similar solutions in the
same coordinate system.
The introduction of the hyperbolic coordinates and the
analysis on hyperbolically
radial waves represent the first step in a longer term program to further our
understanding of PDEs involving the hyperbolic Laplacian operator.

\smallskip

  In the rest of the paper radial always means hyperbolically radial.

\section{Weak solutions and Compatibility Conditions}

In the following, we give the compatibility conditions satisfied by weak solutions following the standard argument. Let us denote by  $\square_{xy}:= \partial_x^2 - \partial_y^2$ and for $(x, y) \in \mathbb R^2$  suppose $u(t, x, y)$ is a solution of (3HNLS) with $\gamma=1$; ie.  $u$ satisfies:
\begin{equation} \label{3HNLS}
i \partial_t u + \square_{xy} u + |u|^2 u = 0, \end{equation}
in the distribution sense and is smooth except on a surface $\Gamma \subset \RR^{1+2}$, let us denote by $\Omega_\pm$ the two domains separated by $\Gamma$ and by $N = (N_0, N_1, N_2)$ the unit outward normal vector of $\Omega_+$. Let $[f] = f_+ - f_-$ where $f_\pm$ is the one side limit of any quantity $f$ on $\Gamma$. For any test function $\phi \in C_0^\infty (\RR^{1+2})$ one may compute that,
\begin{equation} \label{C1}
\int_\Gamma ((0, \nabla \phi) \cdot R N) [u] - \phi R N \cdot (i[u], [\nabla u]) dS =0
\end{equation}
where $\nabla$ denotes the spatial gradient and the matrix $R:={\rm diag}(1, 1, -1)$ performs a reflection across the
$xt$-plane. By taking $\phi$ identically zero on $\Gamma$, we obtain that, if $u$ has a jump discontinuity along $\Gamma$, then
\begin{enumerate}
\item [(C1)] $N_1^2= N_2^2$ or equivalently $\Gamma$ is parallel to $x=y$ or $x=-y$ which are characteristics of $\square_{xy}$.
\end{enumerate}
Therefore, the surface takes the form of
\[
\Gamma = \{ (t, x_0(t) + \alpha_1 s, y_0(t) + \alpha_2 s \mid t,\, s \in \RR\} \text{ and } N_0 = N_0(t)
\]
where $\alpha_{1, 2} = \pm 1$. For such smooth $\Gamma$, integrating by parts on the first integrand leads to
\begin{equation} \label{C2}
\int_\Gamma  \phi R N \cdot (i[u], 2[\nabla u]) dS =0
\end{equation}
and thus we have
\begin{enumerate}
\item [(C2)] $\frac i2 [u] N_0 + [u_x] N_1 - [u_y] N_2=0$.
\end{enumerate}
Condition (C2) is equivalent to that $[u_x]$, $[u_y]$, $i[u]$ are colinear in the complex plane. In fact, if $u$ is piecewise $C^1$ and $[u]=0$, then, on the one hand, (C2) implies that the vector $(N_1, -N_2)$ lies in the kernel of the $2\times 2$ matrix $[\nabla u]$. On the other hand, if $[\nabla u]\ne0$, then its kernel is given exactly by the tangent space of $\Gamma$, i.e. the one dimensional subspace spanned by $(N_2, -N_1)$, and thus $N_1^2 =N_2^2$ holds. This means that, even though (C1) follows from the assumption $[u]\ne 0$, it still follows from (C2) if $[u]=0$ and $[\nabla u]\ne0$. Therefore, (C1) and (C2) should be imposed on weak solutions as long as $u$ or $\nabla u$ have a jump discontinuity along $\Gamma$. Condition (C2) can also be viewed as a necessary condition for the singularity to propagate as a surface in $\RR^{1+2}$. If (C2) is not satisfied by some given initial data, then the singularity will not maintain a simple geometry for $t>0$.

Notice that, according to (C1), the jump discontinuity can occur along lines in the directions of $(1, 1)$ or $(1, -1)$. To study the local behavior of the weak solutions at the intersection of its lines discontinuity, we consider the case where
\[
\Gamma_* = \{(x-x_0(t))^2= (y - y_0(t))^2\},
\]
i.e. a cross in the $xy$-plane with a moving center $c(t) = (x_0(t), y_0(t))$. We will use $u_{\he, \ve}^\pm$ to denote the value of $u$ in the cone either in the horizontal or vertical directions, opening in the positive or negative directions. In particular, $u_{\he, \ve}^{1, 2, 3, 4}$ denote the limiting values of $u_{\he, \ve}^\pm$ on $\Gamma$ in each quadrant. We use $N$ to denote the unit outward normal of the horizontal cones $(x- x_0(t))^2 > (y- y_0(t))^2$. In this case, instead of \eqref{C2}, integrating by parts on the first integrand in \eqref{C1} implies
\begin{align}
4 \int \phi(c(t)) \{u_{\he}^1(c(t))- u_{\ve}^1(c(t)) &+ u_{\he}^2(c(t)) - u_{\ve}^2(c(t)) + u_{\he}^3(c(t)) - u_{\ve}^3 (c(t)) \notag \\
&+ u_{\he}^4(c(t)) - u_{\ve}^4(c(t))\} dt - \int_\Gamma  \phi R N \cdot (i[u], 2[\nabla u]) dS =0. \label{C3}
\end{align}
Therefore, in addition to (C2), we also obtain
\begin{enumerate}
\item [(C3)] $u_{\he}^1- u_{\ve}^1 + u_{\he}^2 - u_{\ve}^2 + u_{\he}^3 - u_{\ve}^3 + u_{\he}^4 - u_{\ve}^4 =0$ along $c(t)$, $t \in \RR$.
\end{enumerate}

\begin{remark}  When a solution has singularity along a surface, the compatibility conditions (C1--C3) above are only {\it necessary}. In fact, one should be careful when talking about the propagation of the singularity as the system does not have finite propagation speed, which can be seen from the formula of the fundamental solution if the free equation \cite{GHS1}. For example, if the initial data has some simple jump discontinuity, this singularity might become unpredictable for $t>0$ except for highly symmetric cases. In what follows $\Gamma$ will describe the cross $X$ in $\RR^2$  defined by  $\{\,|x|= |y|\, \}$.   We will see in Section 4 that the simplest example of data satisfying the compatibility conditions described above are
those of the form $(f, g, -f, -g)$ given in each of the four
cones of $\RR^2 \setminus X$.
\end{remark}

\section{Hyperbolic coordinates and radial solutions}

Notice that the regular Euclidean rotation does not keep $\square_{xy}$ invariant, and hence the usual polar coordinates may not be
the most appropriate ones for this
system. Instead, we introduce an alternative: we foliate $\Bbb R^2\setminus X$ by hyperbolas and consider the following change of variables:

\begin{align}  \begin{cases} \label{coordinates}  x &:= a \cosh \alpha \qquad y:= a \sinh \alpha \quad \quad a^2 := x^2 - y^2 \quad \quad \mbox{when} \quad |x| > |y| \\
  x &:= a \sinh \alpha \qquad y:= a \cosh \alpha \quad \quad a^2 := y^2 - x^2 \quad \quad  \mbox{when} \quad |x| < |y| \end{cases} \end{align}
where $a, \alpha \in \Bbb R,  a\ne 0$.

Under this change of variables we have that the area form $dx dy = |a| da d\alpha$. Moreover  \text{ for any }\,  $a \in \Bbb R\setminus\{0\}$ we have:

\begin{align*} \quad &\nabla a = ( \, \frac{x}{a}, - \frac{y}{a} \,) \qquad \mbox{and} \qquad \square a = \frac{1}{a}  \qquad \quad  \, \, \,  \text{ for}  \quad |x| > |y|, \\
& \nabla a = ( -\frac{x}{a}, \,  \frac{y}{a}\, )    \qquad \mbox{and} \qquad \square a =  - \frac{1}{a} \qquad \quad \text{ for }  \quad |x| < |y|.
\end{align*}

For $\nabla \alpha$ we need to distinguish whether $a>0$ or $a <0$. We have:

\begin{align*} \quad & \nabla \alpha = ( - \frac{y}{a^2}, \, \frac{x}{a^2} \,)  \text{ for  }  \quad |x| > |y| \text{ and }  \, a> 0 \\
\quad & \nabla \alpha = ( \, \frac{y}{a^2}, - \frac{x}{a^2} \,)  \text{ for  }  \quad |x| > |y| \text{ and }  \, a< 0 \\
\quad & \nabla \alpha = (\,  \frac{y}{a^2}, - \frac{x}{a^2} \,)  \text{ for  }  \quad |x| < |y| \text{ and }  \, a> 0 \\
\quad & \nabla \alpha = ( - \frac{y}{a^2}, \, \frac{x}{a^2} \,)  \text{ for  }  \quad |x| < |y| \text{ and }  \, a < 0 \\
\quad &  \square \alpha = 0 \qquad \text{ in all cases. }\end{align*}

We moreover note that
\begin{align*} \quad & a_x^2 - a_y^2 = 1  \qquad \quad \mbox{and} \qquad   \alpha_x^2 - \alpha_y^2 =  - \frac{1}{a^2}  \qquad \text{ for  }  \quad |x| > |y|, \\
 \quad & a_x^2 - a_y^2 =  - 1   \qquad \mbox{and} \qquad     \alpha_x^2 - \alpha_y^2 =   \, \frac{1}{a^2}   \qquad \text{ for  }  \quad |x| < |y|.
\end{align*}

 \medskip

By changing coordinates $(x,y) \longrightarrow (a, \alpha)$ in ${\mathbb R}^2 \setminus X$ the (3HNLS) equation transforms itself into the system {\footnote{ Note that in the first quadrant $a>0$ and $\alpha >0$; in the second quadrant $a<0$ and $\alpha <0$; in the third $a<0$ and $\alpha >0$ and in the fourth quadrant $a>0$ and $\alpha<0$} }

\begin{align}  \begin{cases} \label{system}  i \partial_t u + \partial_{aa} u + \frac{1}{a} \partial_a u - \frac{1}{a^2} \partial_{\alpha \alpha} u + |u|^2 u = 0 \qquad \text{ for } \,  |x| > |y| \qquad \rm{(F)} \\
 i \partial_t u - \partial_{aa} u - \frac{1}{a} \partial_a u + \frac{1}{a^2} \partial_{\alpha \alpha} u + |u|^2 u = 0 \qquad \text{ for } \,  |x| < |y| \qquad \rm{(D)} \end{cases} \end{align}

%
% \section{Radially symmetric solutions}
%

It is clear that \eqref{system} is invariant under the translation in $\alpha$, or equivalently (3HNLS) is invariant under the rotation in the hyperbolic coordinates. Thus it motivates us to study radial or equivariant solutions. We will focus on radial solutions $u(t, a, \alpha)= u(t, a)$. (See Remark \ref{equivariant} for equivariant solutions.)
%
%The above hyperbolic coordinates are useful in studying {\it radial} solutions in the sense of $u(a, \alpha)= u(a)$. This type of radial symmetry is preserved in time evolution.
%
In the rest of this note, we will consider radially symmetric solutions assuming $u_{\he, \ve}^\pm$ is {\it a priori} smooth for $a \in [0, \infty)$ except for some jump discontinuity across $X$. Note in what follows these facts will be used without any further mention. In other words we will be working with

\begin{align}  \begin{cases} \label{system2}  i \partial_t u + \partial_{aa} u + \frac{1}{a} \partial_a u + |u|^2 u = 0 \qquad \text{ for } \,  |x| > |y| \qquad \rm{(F)} \\
 i \partial_t u - \partial_{aa} u - \frac{1}{a} \partial_a u + |u|^2 u = 0 \qquad \text{ for } \,  |x| < |y| \qquad \rm{(D)} \end{cases} \end{align}

The compatibility condition (C3) implies that at the origin $(0,0) \in \Bbb R^2$
\begin{equation} \label{comp}
u_{\he}^+ (0) + u_{\he}^-(0) = u_{\ve}^+(0) + u_{\ve}^-(0).
\end{equation}
Some obvious solutions of this type are $(u_{\he}^-=-u_{\he}^+, u_{\ve}^+=u_{\ve}^-= 0)$ or $(u_{\he}^+=u_{\he}^-=0, u_{\ve}^-=-u_{\ve}^+)$ where $u_{\he}^+$ or $u_{\ve}^+$ solves the above (F) or (D) respectively.

\bigskip

% \section{ Radial Standing Wave Solutions: \, Analysis}

We are interested in the existence (construction) of radial standing wave solutions $u(t, a) = e^{-i \mu t} v(a)$ such that as $a \to  \pm \infty$ $v \to v_{\pm \infty}$ with $v_{\pm \infty}$ a constant.

\begin{remark} \label{va=0}
If $u(x, y)$ is $C^1$ and (hyperbolically) radially symmetric, then the radial symmetry implies $D u(0, 0)=0$ and thus $v_a(0)=0$.
% -appropriately interpreted-.
{\it We assume the latter from now on.}
\end{remark}

From \eqref{system2} we then have that $v$ satisfies:

\begin{align} \begin{cases} \label{standing} \quad  \mu v + v_{aa} +  \frac{1}{a} v_a +  |v|^2 v =0 \quad  \qquad \text{ for } \, |x| > |y| \quad \qquad  \rm{(F)} \\
\quad \mu v - v_{aa} -  \frac{1}{a} v_a +  |v|^2 v =0 \quad  \qquad \text{ for } \, |x| < |y| \quad \qquad  \rm{(D)} \end{cases}\end{align}

\smallskip

\begin{remark}  We call $v_d$ the solution to (D) and $v_f$ the solution to (F) and $v=(v_d, v_f)$ the solution to \eqref{standing}.
Note that a radial solution to \eqref{standing} is $C^1$ if and only if $v_d$ and $v_f$ have equal values at $a=0$;  ie.  $v_d(0)=v_f(0)$.
\end{remark}

\smallskip

\begin{remark} \label{real}
Under the assumption $v_a(0)=0$, without any loss of generality we consider $v$ real valued only. In fact, suppose $v(0)= \omega_0 e^{i\theta_0}$ with 
$omega_0 \in  \mathbb R$ and
$\theta_0 \in \mathbb R$. Let $\tilde v(a) = v(a) e^{-i \theta_0}$, then $\tilde v(0), \, \tilde v_a(0) \in \R$. Theorem \ref{singularODE} implies the solutions $\tilde v(a)$ exist and are real.
\end{remark}

\smallskip

\begin{remark} In the analysis below we only consider $a >0$ but corresponding results hold for $a<0$.
\end{remark}

\smallskip

\begin{remark}
By considering another coordinate change $x' = a \cos \alpha$ and $y' = a \sin \alpha$, \eqref{standing} (F)(D) simply  become $\pm \Delta_{x'y'} v +\mu v + v^3=0$. Solutions of this equation in $L^2$ have been studied extensively (see, for example, \cite{BL1, BL2, BLP, Pohozaev, Strauss}). Based on those results, it is obviously impossible to find solutions of (F) and (D) with $L^2$ decay at infinity so that they match at $a=0$. Instead, we consider all possible bounded solutions here through a more careful analysis.
\end{remark}

\bigskip

\begin{theorem}[Main Theorem 1] \label{maintheorem} $C^1$ weak radial standing wave solutions $u(t, a, \alpha) =
e^{-i\mu t} v(a)$ to \eqref{3HNLS}, which are $C^2$ except along $x^2= y^2$ (or equivalently $a=0$), exist only if $\mu <0$. In fact, for any $v_0 \in (0, \sqrt{-\mu})$, there exists a unique $C^1$ such solution $(u_f, u_d)=(e^{-i\mu t} v_f(a), e^{-i\mu t}v_d(a))$ such that $v_f$ and $v_d$ are smooth on $a\in [0, \infty)$ and
\[
v_d(0)= v_f(0) =v_0 \text{ and } \p_a v_d(0)=0 = \p_a v_f(0).
\]
Moreover, for $a\to \infty$,
\[\begin{split}
&v_f(a) = \sqrt{-\mu} + \frac {r_f}{\sqrt{a}} \cos (\gamma a + \sigma \log a + \beta) + O(a^{-\frac 32}) \\
&v_d(a) = \frac {r_d}{\sqrt{a}} \cos (\gamma a + \sigma \log a + \beta) + O(a^{-\frac 32})
\end{split}\]
for some $\gamma, \beta, \sigma \in \mathbb R$ and $r_f \ne0$ and $r_d\ne 0$.
\end{theorem}

For $v_0 \in (-\sqrt{-\mu}, 0)$, we only need to consider $-v$ and apply the above theorem.

\medskip

Notice that equation (3HNLS) is invariant under the scaling transformation $u \, \to \, u_\lambda (t, x, y) = \lambda u(\lambda^2 t, \lambda x, \lambda y)$. Accordingly, in the hyperbolic coordinates, consider \eqref{system} in the usual similarity coordinates $u(t, a, \alpha) = t^{-\frac 12} \tilde u(\log |t|, t^{-\frac 12}a, \alpha)$, we obtain
\begin{align}  \begin{cases} \label{self-similar}
i \partial_t \tilde u - \frac i2 \tilde u -\frac i2 a \p_a \tilde u + \partial_{aa} \tilde u + \frac{1}{a} \partial_a \tilde u - \frac{1}{a^2} \partial_{\alpha \alpha} \tilde u + |\tilde u|^2 \tilde u = 0 \qquad \text{ for } \,  |x| > |y| \qquad \rm{(F)} \\
i \partial_t \tilde u - \frac i2 \tilde u -\frac i2 a \p_a \tilde u - \partial_{aa} \tilde u - \frac{1}{a} \partial_a \tilde u + \frac{1}{a^2} \partial_{\alpha \alpha} \tilde u + |\tilde u|^2 \tilde u = 0 \qquad \text{ for } \,  |x| < |y| \qquad \rm{(D)}
\end{cases} \end{align} where now $\tilde u = \tilde u( t, a, \alpha)$.
Next we look for  radial standing waves in the self-similar variables of the form $\tilde u = e^{-i \mu t} v(a)$ which satisfies
\begin{align} \begin{cases} \label{SSstanding}
\quad  \mu v - \frac i2 v -\frac i2 a  v_a + v_{aa} +  \frac{1}{a} v_a +  |v|^2 v =0  \qquad \text{ for } \, |x| > |y|  \qquad  \rm{(F)} \\
\quad \mu v - \frac i2 v -\frac i2 a  v_a - v_{aa} -  \frac{1}{a} v_a +  |v|^2 v =0  \qquad \text{ for } \, |x| < |y|  \qquad  \rm{(D)} \end{cases}\end{align}
The remarks we made for the radial standing waves are valid for radial self-similar waves except here we can no longer assume $v$ has only real values.

\bigskip

\begin{theorem}[Main Theorem 2] \label{maintheorem2} For any $\mu \in \R$, there exists $v_*>0$ such that for any $v_0 \in (-v_*, v_*)$, there exists a unique $C^1$ weak  radial self-similar wave solution $(u_f, u_d)$ of \eqref{3HNLS}
which is $C^2$ except at $x^2 = y^2$ (or equivalently $a=0$) and
\[
u_f(t, 0)=u_d(t, 0)=v_0 t^{-\frac 12} e^{-i\mu \log |t|} \text{ and } \partial_a u(t, 0)=0
\]
As $\frac a{\sqrt{t}} \to \infty$, there exist constants $r_f\ne 0$, $r_d\ne 0$, $\theta_f$, and $\theta_d$ such that
\[\begin{split}
&u_f(t, a) = \frac {r_f}a e^{i(\frac {a^2}{8t} -\mu \log |t|)} \cos (\theta_f - \frac {a^2}{8t} - 2\mu \log \frac a{\sqrt{t}})  + O(\frac {t^{\frac 32}}{a^3})\\
&u_d(t, a) = \frac {r_d}a e^{i(-\frac {a^2}{8t} -\mu \log |t|)} \cos (\theta_d - \frac {a^2}{8t} + 2\mu \log \frac a{\sqrt{t}})  + O(\frac {t^{\frac 32}}{a^3}).
\end{split}\]
%\[\begin{split}
%&v_f (a) = \frac {r_f}a e^{\frac i8 a^2} \cos (\theta_f - \frac {a^2}8 - 2\mu \log a)  + O(a^{-3})\\
%&\p_a v_f (a) = \frac {r_f}{4a} e^{\frac i8 a^2} \sin (\theta_f - \frac {a^2}8 - 2\mu \log a)  + O(a^{-3})\\
%&v_d (a) = \frac {r_d}a e^{-\frac i8 a^2} \cos (\theta_d - \frac {a^2}8 + 2\mu \log a)  + O(a^{-3})\\
%&\p_a v_d (a) = \frac {r_d}{4a} e^{-\frac i8 a^2} \sin (\theta_d - \frac {a^2}8 + 2\mu \log a)  + O(a^{-3})
%\end{split}\]
\end{theorem}

\medskip

\begin{remark}
From the asymptotic form of the above self-similar waves, for any $a\ne0$, i.e. away from the cross $|x|=|y|$, as $t$ approaches the singular value $0$ (here $t$ may be replaced by $\pm (t-t_0)$), the complex values of the solutions oscillate rapidly in both the angle and the magnitude. In particular, their magnitudes oscillate with an envelope proportional to $\frac 1a$. At the cross $|x|=|y|$, the magnitude of the solutions blows up like $\frac 1{\sqrt{t}}$. Therefore they strongly display a kind of profile of the X shape.
\end{remark}

\medskip

To prove these theorems, we first prove existence of solutions near $a=0$ and next we study their properties as $a \to +\infty$. Equations \eqref{standing}(F)(D) and \eqref{SSstanding}(F)(D) have been extensively studied. To establish local existence of solutions near $a=0$,  well known methods include for example, power series or variational methods. Here we present (in the Appendix) another method based on invariant manifold theory to study solutions near a regular singular point of a class of ODE system of which \eqref{standing}(F)(D) and \eqref{SSstanding}(F)(D) are special cases. The advantage of this method is that it does not require either much smoothness as in the power series approach or variational structures. More specifically, for any $v_0$, let $x_1=v-v_0$ and $x_2 = av_a$ in these four ODEs. Each of them turns into a form for which we may apply Theorem \ref{singularODE}. For example, for \eqref{standing}(F), in the notation of Theorem \ref{singularODE} we have that $F(a, x, v_0) = (x_2, -\mu a^2 (x_1 + v_0) - a^2 (x_1 + v_0)^3)^T$ with the parameter $v_0$ and  $F(0, 0, v_0) =0$ for all $v_0$. Moreover $F_x (0, 0, v_0)= \begin{pmatrix} 0 &1 \\ 0 &0\end{pmatrix}$ and thus $E_+ =\{0\}$. Therefore, for any $v_0$, there exists a unique solution $v(a)$  for small $a\ge 0$ so that
%%CCZ
$v$ and $a v_a$ are  H\"older continuous in both $a$ and $v(0)= v_0$. Moreover, the statement (4) in Theorem \ref{singularODE} implies $x_2 = a v_a = O(a^2)$ which implies $v_a(0)=0$. The other equations are handled similarly.
\medskip

\begin{remark} \label{equivariant}
The same Theorem \ref{singularODE} can also be applied to study equivariant standing waves or equivariant self-similar solutions where the equivariance means $u(t, a, \alpha) = e^{ic \alpha} e^{-i \mu t} v(a)$, $c, \mu \in \R$. In this case, for equivariant standing waves, \eqref{system} implies
\begin{align*} \begin{cases}
\quad  \mu v + v_{aa} +  \frac{1}{a} v_a + \frac {c^2}{a^2} v +  |v|^2 v =0 \quad  \qquad \text{ for } \, |x| > |y| \quad \qquad  \rm{(F)} \\
\quad \mu v - v_{aa} -  \frac{1}{a} v_a  - \frac {c^2}{a^2} v +  |v|^2 v =0 \quad  \qquad \text{ for } \, |x| < |y| \quad \qquad  \rm{(D)}
\end{cases}\end{align*}
For $c\ne 0$, Theorem \ref{singularODE} immediately implies that the only continuous solution on $[0, a_0)$, $a_0>0$, of either of the above has to be trivial. Therefore, there are no nontrivial continuous equivariant standing waves. The same argument applies to the self-similar case as well.
\end{remark}

To finish the proof of the theorems, we will study the asymptotic properties of these systems as $a \to \infty$ in the next two sections.

\bigskip

\section{Radial Standing Wave Solutions}

Note that by considering $a_1 = |\mu|^{\frac 12} a$ and $v = |\mu|^{\frac 12} v_1$, we only need to consider $\mu = \pm1, 0$ for \eqref{standing}(F)(D). To study the asymptotic properties of the solutions as $a \to +\infty$ we proceed according to the sign of $\mu$.

\subsection*{ Case $\mu > 0$}

\medskip

First we show that all solutions $v_f$ to \eqref{standing} (F)  go to zero  like $\frac{1}{\sqrt{a}} \cos( \gamma  a + \sigma \log a + \beta)$
as $a \to \infty$.
Indeed, for nontrivial $v$ satisfying \quad $v_{aa} + \frac{1}{a} v_a  + \mu v + |v|^2 v =0$ and $a >0$ define $v =: \frac{1}{\sqrt{a}} w$. Then $w$ satisfies:

\begin{equation} \label{w-equation1} (\frac{1}{4a^2} + \mu) \frac{1}{\sqrt{a}} w +  \frac{1}{\sqrt{a}} w_{aa} +  \frac{1}{a}  \frac{1}{\sqrt{a}} |w|^2 w =0.
\end{equation}

We want to
first show that $w$ is bounded.  \,   Multiply \eqref{w-equation1} by $\sqrt{a}$ to get

\begin{equation}  \label{w-equation2} w_{aa} + (\frac{1}{4a^2} + \mu)  w +  \frac{1}{a} |w|^2 w =0.  \end{equation} Multiplying \eqref{w-equation2} by $w_a$ and rewriting the resulting expression
using `perfect derivatives' we get

\begin{equation}\label{w-equation3} \partial_a \biggl[   \frac{1}{2} (w_a^2) +    \frac{1}{8 a^2}  w^2  + \frac{\mu}{2} w^2 +  \frac{1}{4a} w^4 \biggr]  \, =\, -\frac{1}{4 a^3} w^2 - \frac{1}{4a^2} w^4. \end{equation}
 \bigskip

Define  \begin{equation} \label{energy1} E(a):=  \frac{1}{2} (w_a^2) +    \frac{1}{8 a^2}  w^2  + \frac{\mu}{2} w^2 +  \frac{1}{4a} w^4 \geq 0. \end{equation} Then we have

\begin{equation}\label{energy2} \partial_a E(a)  \, =\, -\frac{1}{4 a^3} w^2 - \frac{1}{4a^2} w^4 \leq 0 \end{equation}
Hence $E(a)$ is decreasing, which implies $w$ and $w_a$ are bounded and
$E_0:=\lim_{a\to \infty} E(a)$ exists. The boundedness of $w$ and \eqref{energy2} again imply
\[
0 \ge \partial_a E \ge - \frac C{a^2}  E
\]
for some $C>0$ independent of $a$. Therefore $E_0>0$ and $E(a) - E_0 = O(\frac 1a)$.  Moreover,  let us
change the variables
\[
w= r \cos \theta, \quad w_a = r \sin \theta.
\]
Recall $\mu =1$, we have $r^2 - 2 E_0 = O(\frac 1a)$ and thus
\begin{equation} \label{r}
r = \sqrt{2E_0} + O(\frac 1a)
\end{equation}
as $a \to \infty$. One may compute
\[\begin{split}
r_a =& -\frac 1{4a^2} r \cos \theta \sin \theta - \frac 1a r^3 \cos^3 \theta \sin \theta\\
\theta_a =& -1 - \frac 1{4a^2} \cos^2 \theta - \frac 1a r^2 \cos^4 \theta\\
=& -1 - \frac {3E_0}{8a} + \frac {E_0}a (\frac 38 - \cos^4 \theta) + \frac {E_0 - r^2}a \cos^4 \theta.
\end{split} \]
Integrating the $\theta$ equation by parts and using $\theta_a = -1+ O(\frac 1a)$, one may show,
\begin{equation} \label{theta}
\theta = \theta_0 + a + \frac {3 E_0}8 \log a + O(\frac 1a).
\end{equation}

\bigskip

We next consider solutions $v_d$ to \eqref{standing} (D).   Multiplying (D) by $av$ and integrating on $[0, a]$, we obtain
\[
(v^2)_a = \frac 2a \int_0^a  {\tilde a} \,  ((v_a)^2 + v^2 + v^4) \, \, d{\tilde a}
\]
which implies $v^2 \ge O(\log a)$ as $a \to \infty$ unless $v\equiv 0$. Therefore there are no nontrivial bounded solutions $v_d$ to \eqref{standing}(D).

All in all we conclude then that  for the case $\mu >0$ there are no  bounded solutions other than $(v_f, 0, -v_f, 0)$, which can not be in $C^0$ unless trivial.

\bigskip

\subsection*{ Case $\mu =0$}

For \eqref{standing} (F), let $v= : a^{-\frac 13} w$ one may compute
\[
w_{aa} + \frac 1{3a} w_a + \frac 1{9a^2} w + a^{-\frac 23} w^3=0.
\]
Multiplying it by $a^{\frac 23} w_a$, we have
\[
\p_a (\frac 12 a^{\frac 23} (w_a)^2 + \frac 14 w^4 + \frac 1{18} a^{-\frac 43} w^2) = -\frac 2{27} a^{-\frac 73} w^2 \le 0.
\]
Therefore $w=O(1)$ and $w_a = O(a^{-\frac 13})$ as $a \to \infty$, which implies
\[
v_a = O(a^{-\frac 23}), \quad v= O(a^{-\frac 13}), \text{ as } a \to \infty.
\]
In fact, for different values of $v(0)$, the solution $v(a)$ only differs
by a scaling change.

Finally, the same argument as in the case of $\mu>0$ shows that the only bounded solution $v_d$ to \eqref{standing}(D) is the trivial one.

We thus also conclude in the case $\mu=0$ that there are no bounded solutions other than $(v_f, 0, -v_f, 0)$, which can not be in $C^0$ unless is trivial.

\bigskip

\subsection*{{Case $\mu <0$ ($\mu =-1$)}}  From the uniqueness given in (4) of Theorem \ref{singularODE}
of solutions which are H\"older continuous at $0$, clearly $v_d, v_f \equiv 0,\,\pm1$ if $v_d(0), v_f(0) =0, \,\pm1$.

We first consider solutions $v_d$ to \eqref{standing} (D) and claim that if  $|v_d (0)|>1$, then $|v_d(a)| >1$ and $|v_d(a)| \to \infty$ as $a \to \infty$. Indeed, without any loss of generality assume $v_d(0)>1$ and $v_d(a)= 1+z(a)$. One may compute
\begin{equation} \label{zd}
z_{aa} + \frac 1a z_a - 2z - 3z^2 - z^3 =0.
\end{equation}
Multiplying it by $z_a$, we have
\[
\p_a (\frac 12 (z_a)^2 - (z^2 + z^3+ \frac 14 z^4)) = -\frac 1a (z_a)^2,
\]
which implies
\[
0>- (z(0)^2 +z(0)^3+ \frac 14 z(0)^4) \ge (\frac 12 (z_a(a))^2 - (z(a)^2 + z(a)^3+ \frac 14 z(a)^4)).
\]
Therefore, $z(a) \neq 0$,  for all $a>0$ which yields $z(a)>0$. Multiplying \eqref{zd} by $-a z$ implies
\[
\p_a (z^2) = \frac 2a \int_0^a {\tilde a} \, ((z_a)^2 + 2z^2 + 3z^3+ z^4) \, \, d{\tilde a}
\]
which implies $z^2 \ge c\log a$. \\

Thus we assume $v_d(0) \in (-1, 1)$ and claim  $v_d \to 0$ like $\frac{1}{\sqrt{a}} \cos( \gamma  a + \sigma \log a + \beta)$  as $a \to +\infty$.  Indeed, multiplying \eqref{standing} (D) by $v_a$ we obtain
\[
\p_a (\frac 12 (v_a)^2 + \frac 12 v^2 - \frac 14 v^4) = - \frac 1a (v_a)^2 \le 0.
\]
Since $\frac 12 v^2 - \frac 14 v^4$ is strictly increasing in $v^2$ for $v^2 \in [0, 1]$, we have
\begin{equation} \label{1}
v(0) \in (-1, 1) \implies v(a)^2 \le v(0)^2 <1, \quad a>0.
\end{equation}
We will follow closely the procedure which we used to handle solutions to \eqref{standing} (F) in the case of $\mu >0$. Define $v = : \frac{1}{\sqrt{a}} w$. Then $w$ satisfies:
\begin{equation}\label{weq1}  w_{aa} + (\frac{1}{4a^2} + 1)  w - \frac{1}{a} w^3 =0 \end{equation}
Multiplying \eqref{weq1} by $w_a$ and rewriting the resulting expression as `perfect derivatives' we get
\begin{equation}\label{weq2}  \partial_a \biggl[   \frac{1}{2} (w_a^2) +    \frac{1}{8 a^2}  w^2  + \frac1{2} w^2 -  \frac{1}{4a} w^4 \biggr]  \, =\, -\frac{1}{4 a^3} w^2 + \frac{1}{4a^2} w^4. \end{equation}
Define now
\begin{equation}\label{energyb1} E(a):=  \frac{1}{2} (w_a^2) +    \frac{1}{8 a^2}  w^2  + \frac1{2} w^2 -  \frac{1}{4a} w^4  \ge \frac{1}{2} (w_a^2) + (\frac 12 - \frac {|v(0)|^2}4) w^2, \end{equation}
where we used $|a^{-\frac 12} w| = |v| \le |v(0)|$. From \eqref{weq1} we have that
\begin{equation}\label{energyb2}    \partial_a E(a)  \, =\, -\frac{1}{4 a^3} w^2 + \frac{1}{4 a^2} w^4 \le  (\frac {w^2}{4a^2}) w^2.  \end{equation}
Using $|a^{-\frac 12} w| = |v| \le |v(0)|$ and \eqref{1} again, we obtain from \eqref{energyb2}
\begin{equation} \label{diffineq}
\partial_a E(a) \le \frac ca E, \quad c = \frac {|v(0)|^2}{2 - |v(0)|^2} <1. \end{equation}

The differential inequality \eqref{diffineq} implies $E \le O(a^{c})$ and thus $w\le O(a^{\frac c2})$ as $a\to \infty$.  Now from \eqref{energyb1} and \eqref{1} we obtain that
\begin{equation}\label{energyb3}
E(a) \geq \kappa_0 w^2 \quad  \mbox{ where } \quad  \frac{1}{2}>  \kappa_0:= \frac 12 - \frac {|v(0)|^2}4 > 0.
\end{equation}
Using \eqref{energyb2} and \eqref{energyb3} we conclude that
\begin{equation}\label{energyb4}  \partial_a E(a) \leq O(a^{c-2}) w^2  \leq O(a^{c-2}) E.  \end{equation} Hence $E=O(1)$, which gives $w =O(1)$ and  $w_a = O(1)$ from \eqref{energyb1}. The oscillatory asymptotic form of $v$ then follows from the exactly same argument used in the case $\mu >0$.

\bigskip

Next, we consider solutions $v_f$ to \eqref{standing} (F).
We will show that for $v_f(0) \in (0, 2)$, solutions $v_f$ to \eqref{standing}(F) behave like $1+ \frac{1}{\sqrt{a}} \cos( \gamma  a + \sigma \log a + \beta)$  as $a \to +\infty$. By the odd symmetry of the equation, the case when
$v_f(0) \in (-2, 0)$ is the same. First, multiplying \eqref{standing} (F) by $v_a$, we have
\[
\p_a (\frac 12 (v_a)^2 + \frac 14 (v^2 -1)^2) = -\frac 1a (v_a)^2
\]
which, along with $v_a(0)=0$ implies that
\begin{equation} \label{v0}
|v(a)^2 -1| \le |v(0)^2 -1|.
\end{equation}
In particular, \eqref{v0} yields $v(a) >0$. Again, let $v =:1+ a^{-\frac 12} w$; then $w$ satistfies
\begin{equation} \label{wf}
%z_{aa} + \frac 1a z_a + 2z + 3z^2 + z^3 =0.
w_{aa} + (2+ \frac 1{4a^2}) w + 3a^{-\frac 12} w^2 + \frac 1a w^3=0.
\end{equation}
Multiplying \eqref{wf} by $w_a$, we have
\begin{equation} \label{paE}
\p_a E = - \frac 14 a^{-3} w^2 - \frac 12 a^{-\frac 32} w^3 - \frac 14 a^{-2} w^4 = - \frac 14 a^{-1} w^2 (a^{-2} + 2a^{-\frac 12} w + a^{-1} w^2)
\end{equation}
where
\begin{equation} \label{E}
E = \frac 12 (w_a)^2 + (1+\frac 18 a^{-2}) w^2 + a^{-\frac 12} w^3 + \frac 14 a^{-1} w^4 \ge \frac 12 (w_a)^2 + \frac 14 (1 + v)^2 w^2 > \frac 12 (w_a)^2 + \frac 14 w^2.
\end{equation}
In the last two inequalities, we used the fact $v(a) >0$. For the same reason, we can estimate the right hand side of of \eqref{paE} as
\[
\p_a E \le \frac {w^2}{4a} (1- v^2) \le \frac ca E, \quad c= |v(0)^2 -1| <1.
\]
It implies $E=O(a^c)$ and thus $w=O(a^{\frac c2})$ as $a \to \infty$. Substituting this into \eqref{paE} we obtain
\[
\p_a E = O(a^{-\frac {3-c}2} )E
\]
and thus $E=O(1)$ and in fact $E \to E_0>0$ as $ a\to \infty$ unless $w \equiv 0$. The oscillatory asymptotic form of $w$ can be obtained much as in the previous cases. \\

In fact, by a more careful analysis, one may show the behavior of $v_f$ at $a \to \infty$ is always one of
oscillatory convergence to $\pm 1$ except maybe for a sequence of data $v_f(0) = v_j$, $|v_j| > 2$ which make
$v_f$ converge to $0$ exponentially. Since we seek bounded continuous solutions to \eqref{standing} and we have already showed that solutions $v_d$ to \eqref{standing}(D) are unbounded if  $|v_d(0)|>1$, the above result is
sufficient for us.
%we only consider $v_f$ with $v_f(0) \in (-1, 1)$. Without any loss of generality assume $v_f(0) \in  (0, 1)$.  We show the behavior of $v_f$ at $a \to \infty$ is always one of oscillatory convergence to $\pm 1$ except maybe for a sequence of data $v_f(0) = v_j$, $|v_j| > 2$  which make $v_f$ converge to $0$ exponentially.  However and as pointed out above the latter situation will not yield any bounded $C^1$ solutions and hence we do not go through it here.

From the above we conclude then that for the case  $\mu <0$ there exist smooth $v_d$ and $v_f$ so that $v_d(0)= v_f(0) \in (-1, 1)$, $\p_a v_d(0)=0 = \p_a v_f(0)$, $v_f(\infty) = \pm 1$ and $v_d (\infty) =0$ where the convergence is like $a^{-\frac 12} \cos (\gamma a + \beta \log a + \sigma + O(\frac 1a))$.

\bigskip

\section{Radial self-similar waves}

To analyze the system \eqref{SSstanding} for self-similar solutions, we first change the independent variable to $s= \frac {a^2}2$ and then change the unknown to $w=s^{\frac 12} e^{\mp \frac i4 s} v$. Consequently, \eqref{SSstanding} becomes
\begin{align} \begin{cases} \label{SSstanding1}
\quad w_{ss} + (\frac 1{16} +\frac 1{4s^2} + \frac \mu{2s}) w + \frac 1{2s^2} |w|^2w =0 \qquad \text{ for } \, |x| > |y|  \qquad  \rm{(F)}\\
\quad w_{ss} + (\frac 1{16} +\frac 1{4s^2} - \frac \mu{2s}) w - \frac 1{2s^2} |w|^2w =0  \qquad \text{ for } \, |x| < |y|  \qquad  \rm{(D)} \end{cases}\end{align}
Without loss of generality (see Remark \ref{va=0} and \ref{real}), we may assume $w$ is real-valued. It is clear that, as long as we obtain the boundedness of the solutions to the ODEs, their oscillatory asymptotic forms follow from the exactly same change of variables as in the previous section $w=r\cos \theta$ and $w_s = \frac 14 r \sin \theta$ and similar arguments. So we focus only on the boundedness.

Multiplying \eqref{SSstanding1} by $w_s$ and letting
\begin{equation}\label{SSenergy}
E(s) = \frac 12 w_s^2 + (\frac 1{32} + \frac 1{8s^2} \pm \frac \mu{4s}) w^2 \pm \frac 1{8s^2} w^4
\end{equation}
we obtain
\begin{equation}\label{SSderivative}
\p_s E(s) = -(\frac 1{4s^3} \pm \frac \mu{4s^2}) w^2 \mp \frac 1{4s^3} w^4,
\end{equation}
which we would like to control in terms of \eqref{SSenergy}.
\subsection*{Focusing equation (F)}

In this case, for  $s > s_0 := 16 (|\mu|+1)$ we have $E>0$ and
\[
\p_s E(s) \le \frac {|\mu|}{4s^2} w^2 \le \frac {16|\mu|}{s^2} E(s).
\]
Therefore $E(s) \to E_0 >0$ as $s\to \infty$ which implies that $w$ and $w_s$ are bounded as desired.

\subsection*{Defocusing equation (D)}

In this case, we follow a standard bootstrap argument.  For those solution $w(s)$ such that $|w(s_0)|\le 1$, let
\[
s_1 := \sup \{s\ge s_0 \mid |w(s')| \le 1 \, \forall s' \in [s_0, s]\}.
\]
For $s \in [s_0, s_1]$, we have from \eqref{SSenergy}

\begin{equation}\label{energybound}
E(s) \ge \frac 12 w_s^2 + \frac 1{64} w^2
\end{equation}
and from \eqref{SSderivative} and \eqref{energybound} we obtain that
  \[\p_s E(s) \le \frac {|\mu|}{4s^2} w^2 + \frac 1{4s^3} w^4 \le \frac {|\mu| +1}{4s^2} w^2 \le \frac {16(|\mu|+1)}{s^2} E(s).
\]
Therefore, if $|w(s_0)|$ and $|w_s(s_0)|$ are sufficiently small -which by the continuity in Theorem \ref{singularODE} is guaranteed by taking $w(0)$ sufficiently small-  we have that $s_1 =\infty$. It follows $w(s)$ and $w_s(s)$ are bounded.

\bigskip

\section{Conclusions \& Future Directions}

In conclusion, in the present work, we have considered the two-dimensional
hyperbolic NLS equation. We have illustrated the relevance of the use
of the so-called hyperbolic variables in the context of this PDE model
(relevant to optical media combining normal dispersion --typically in the
time variable-- to anomalous one --typically in the continuous spatial
variables). We have
discussed the context of weak solutions within the model and how
relevant compatibility conditions naturally arise along the characteristic
lines of the hyperbolic operator. We have subsequently elucidated the
prototypical weak standing wave solutions that the model supports
establishing that they only exist for negative values of the propagation
constant $\mu$ and characterizing their weak modulated power law decay
on the basis of energy-type methods (in the radial hyperbolic variable).
Also, similar techniques but in appropriately rescaled variables
have been used to discuss the existence of radial, self-similar
wave solutions to the PDE for all $\mu \in \R$.

It would be particularly interesting to revisit the numerical computations
of earlier works \cite{CoTri,trillo1,trillo2,DROU,EFR} in the context
of weak solutions and to examine whether for appropriately crafted initial
conditions the relevant direct numerical
computations support the decay rates analytically obtained herein.
Preliminary computations with radial initial data such as $u(x,y,0)=
\tanh(a)$ for $|x|>|y|$ (and zero otherwise) seem to support the
radial evolution of the solution profile and the development of oscillations
around the unit asymptotic state, as suggested by our Main Theorem (see Fig.
\ref{fig1}).   In particular, this initial profile provides a
vanishing solution in the (D) region, as well as a radial solution
asymptoting to unity in the (F) region. The observation that can be
made based on these preliminary results is that in the latter region
the data indeed remains radial in nature (since the hyperbolae forming
in the middle panel are equi-$a$ lines) yet it also develops an
oscillatory dynamics associated with the decay to the stationary
state, which is reminiscent of the oscillatory convergence suggested
by our Main Theorem. It should also be noted that the numerical
method used here was a 4th order Runge-Kutta for marching the system
in time, combined with a centered-difference in space scheme in a
sufficiently fine spatial grid.
However, detailed numerical investigations are relevant
to establish the relevant decay; these will be deferred to
a future publication.
From the point of view of analytical considerations, it would be
relevant to generalize the results obtained herein to the more
experimentally tractable directions of 3-dimensional media
(with two focusing and one defocusing direction) as in
\cite{trillo1,trillo2,CoTri,CoDiTri}. In that setting, it would
be relevant to connect generalizations of the present results to the
rapidly evolving literature on the nonlinear X-waves of the above works.
Another highly promising direction is that
of the quasi-discrete waveguide array setting of equation  (\ref{eqn1}) in
$2+1$-dimensions or of \cite{heinrich} in $3+1$-dimensions.
 Notice, however, that even the setting of equation (3HNLS)
considered herein may be of relevance to physical applications
of light propagation within planar waveguides in glass membrane
fibers, where already some interesting phenomena
such as linear and nonlinear guidance with ultralow optical attenuation
have been observed; see e.g.
\cite{kolesik}.

\begin{figure}
\begin{center}
%\begin{tabular}{ccc}
    \includegraphics[width=\middlefig]{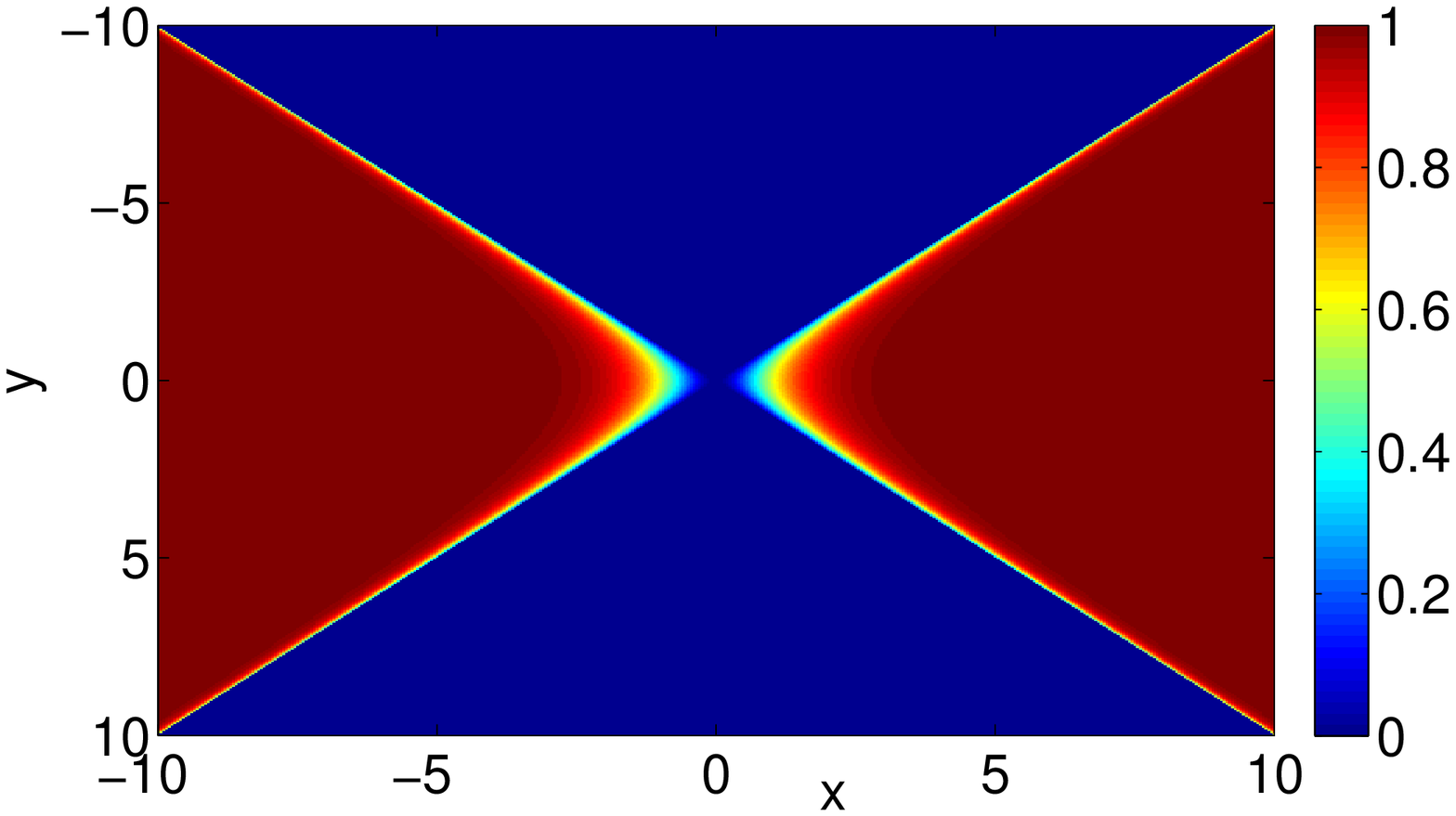}
%&

    \includegraphics[width=\middlefig]{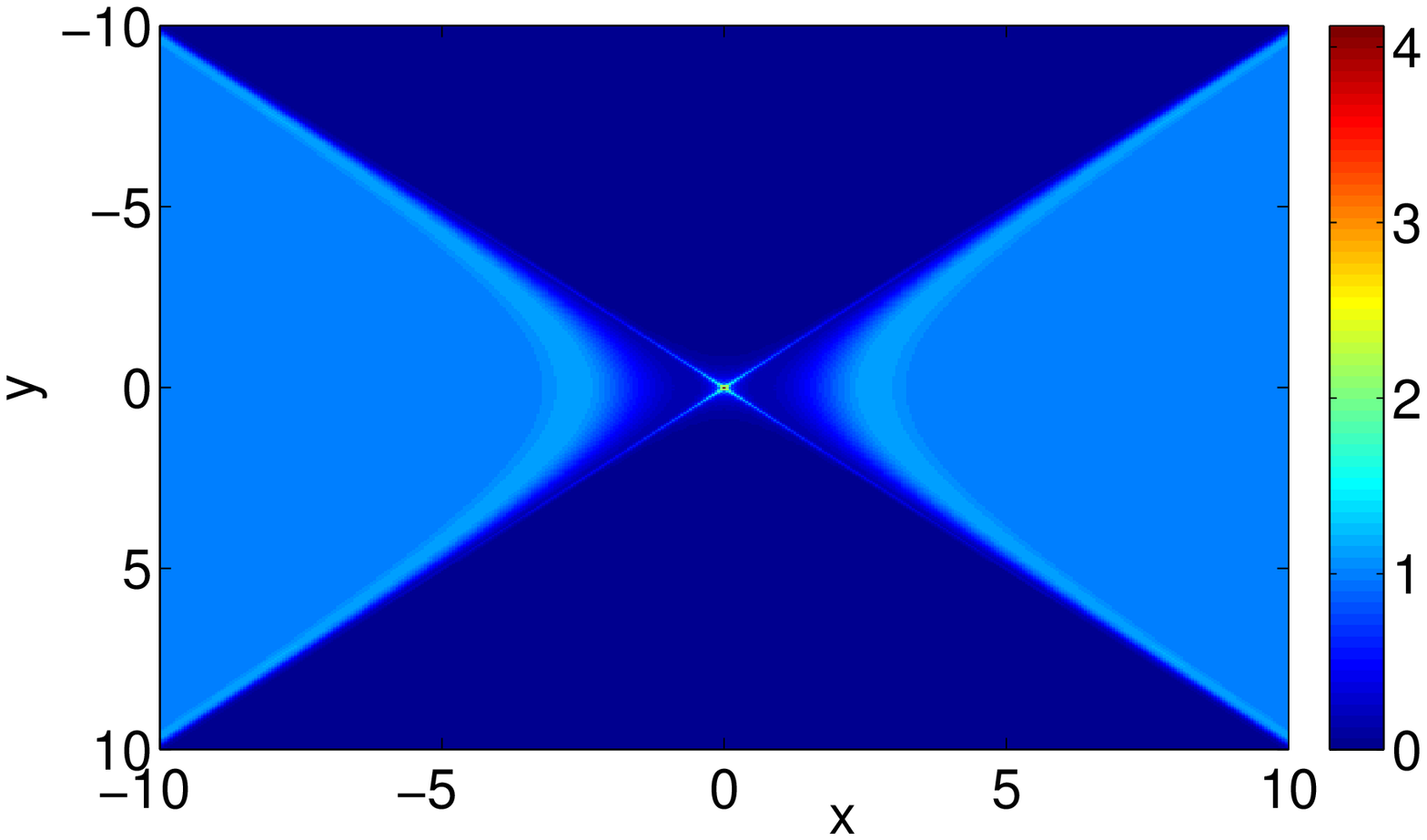}
%&

     \includegraphics[width=\middlefig]{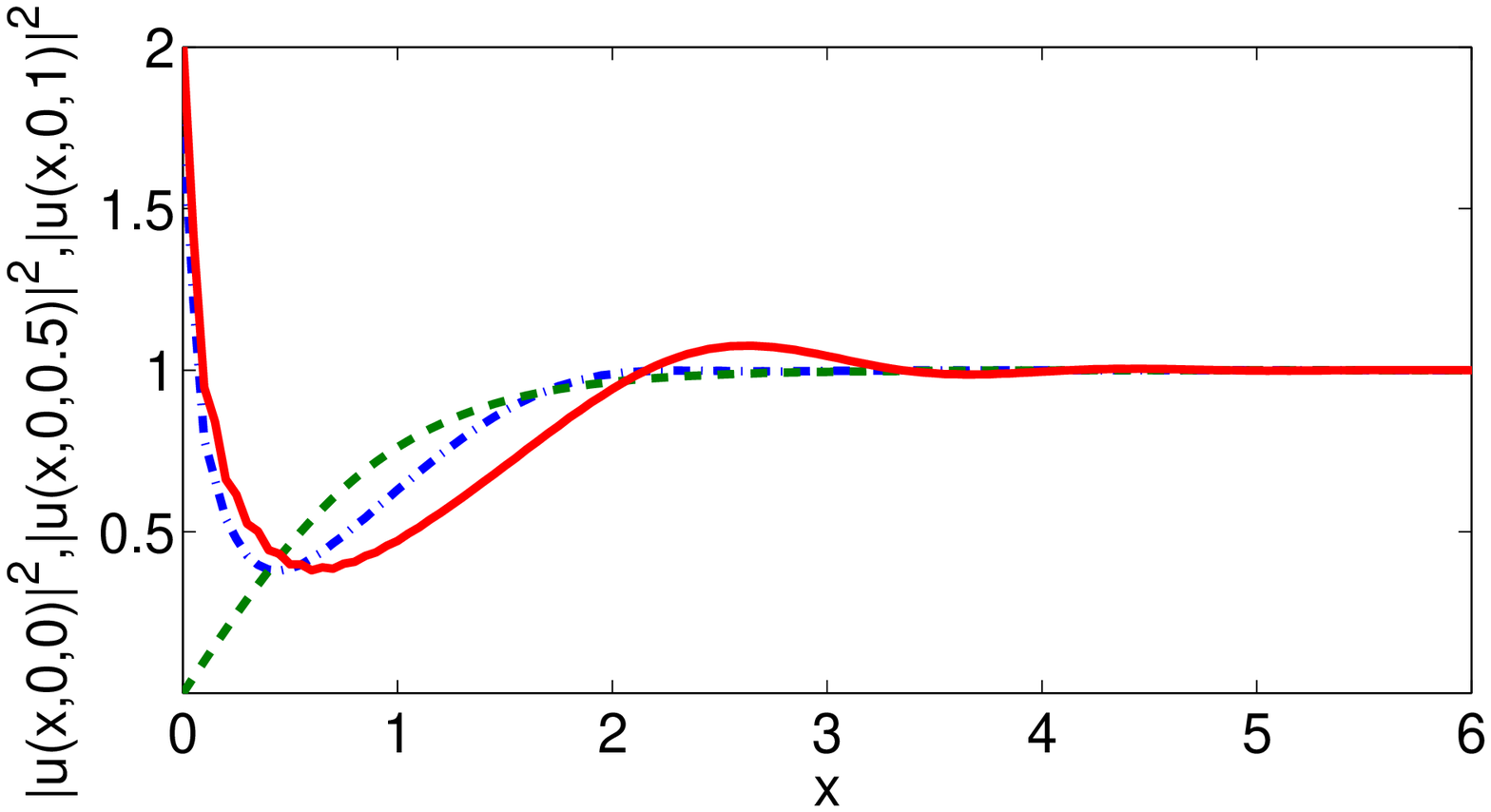}
%\end{tabular}
\caption{Evolution of Eq. (\ref{3HNLS}) for $\gamma=1$ (a factor of
$1/2$ is also used in front of the $\square_{xy}$ as is often done in the
physics literature). The top panel shows the initial condition and the
middle panel shows the result of the evolution at $t=1$. Both panels
show contour plots of $|u|^2$. Notice the hyperbolic
equi-$a$ contours forming
in the middle panel.
The bottom panel
shows a radial cross-section (at $y=0$) of the dynamics, illustrating
the oscillations around the asymptotic state that spontaneously
develop in the dynamics (as suggested by our Main Theorem). The dashed
line shows the initial density profile, the dash-dotted one shows its
evolution at $t=0.5$ and the solid one its evolution at $t=1$.}
\label{fig1}
\end{center}
\end{figure}

\bigskip

\section{Appendix: An ODE blow-up theorem near a regular singular point}

\medskip

%Equations \eqref{standing}(F)(D) have been extensively studied. To establish local existence of solutions near $a=0$,  well known methods include for example, power series or variational methods. Here we present another method based on invariant manifold theory to study solutions near a regular singular point of a class of ODE system of which \eqref{standing}(F)(D) are special cases. The advantage of this method is that it does not require either much smoothness as in the power series approach or variational structures.
%\medskip

Consider
\begin{equation} \label{sODEs}
x' = \frac 1t F(t, x), \quad t>0, \; x \in \R^n, \quad F \in C^k(\R^{n+1}, \R^n).
\end{equation}
Here $t=0$ is a so called regular singular point of the system and we are interested in solutions continuous or smooth at $t=0$. A subcategory of this system is a scalar valued equation
\[
y^{(n)} = \frac 1{t^n} H(t, y, ty', \ldots, t^{n-1} y^{(n-1)})
\]
which include \eqref{standing}(F)(D) and \eqref{SSstanding}(F)(D) as special cases. In fact, by setting $x_j = t^{j-1} (y^{(j-1)} - c_j)$ for $j=1, \ldots, n$ and arbitrary constants $c_1, \ldots, c_n$, the above equation turns into the form of \eqref{sODEs}.

\medskip

Let $B(X, r)$ denote the ball centered at $0$ and with the radius $r$ in a Banach space $X$.

\begin{theorem} \label{singularODE}
Assume $k\ge 1$.
\begin{enumerate}
\item There exists a solution $x(t)$ to \eqref{sODEs} which is continuous on $t \in [0, \varepsilon]$ for some $\varepsilon>0$ and $x(0)= x_0$ if and only if $F(0, x_0)=0$;
\item Suppose $F(0, x_0)=0$. Let $E_+$ be the generalized eigenspace of $F_x(0, x_0)$ corresponding to all eigenvalues whose real parts are greater than $0$. Then there exist $\varepsilon >0$ and a $C^k$ mapping $h :B(E_+, \varepsilon) \to  B((E_+)^\perp, \varepsilon)$ so that: \,  $x(t)$ is a solution of \eqref{sODEs} on $(0, \varepsilon]$ with $x(\varepsilon) \in B(E_+, \varepsilon) + B((E_+)^\perp, \varepsilon)$ and H\"older continuous for $t \in [0, \varepsilon]$ and $x(0) = x_0$ if and only if $x(\varepsilon) \in graph(h)$.
\item The above solutions are of order $O(t^\beta)$ where $\beta \in (0, \beta_0)$ and
\[
\beta_0= \min\{1, \inf\{\rm{Re} \, \lambda  \mid \lambda \text{ is an eigenvalue of } F_x(0, x_0) \text{ and } \rm{Re}\, \lambda>0\}\}.
\]
\item If $k \ge 2$ and $E_+ =\{0\}$, there exists only a unique solution $x(t)$ on $(0, \varepsilon]$ so that it is H\"older continuous in $t$ on $[0, \varepsilon]$ and $x(0)=x_0$. Moreover, if $F_t(0, x_0)=0$, $x(t)=O(t^2)$ as $t \to 0$.
\item If $F=F(t, x, \alpha)$ with a external parameter $\alpha$, $F \in C^k$, and $F(0, x_0, \cdot) \equiv0$, the mapping $h$ in item (2) is also $C^{k-1}$ in $\alpha$.
\end{enumerate}
\end{theorem}

%Here $B(X, r)$ is the ball centered at $0$ and with the radius $r$ in a Banach space $X$.

\begin{proof}
To blow up the system to analyze solutions near $t=0$, we create an auxiliary independent variable $\tau$ so that $\frac {dt}{d\tau}=t$, use $\dot x$ to denote the differentiation in $\tau$, and rewrite \eqref{sODEs} as an autonomous system
\begin{equation} \label{aut} \begin{cases}
\dot t = t\\
\dot x = F(t, x)
\end{cases} \text{ or } \dot {\tilde x} = \tilde F(\tilde x)\end{equation} where $\tilde x= (t, x)^T$.
A solution of \eqref{sODEs} is continuous at $t=0$ is equivalent to that its corresponding solution of \eqref{aut} converges as $\tau \to -\infty$. On the one hand, it is clear that $x_0 = \lim_{\tau \to -\infty} x(\tau)$ only if $\tilde x_0= (0, x_0)^T$ is a fixed point of \eqref{aut}. On the other hand, suppose $F(0, x_0)=0$, we notice $1$ is an eigenvalue of $D\tilde F(\tilde x_0)$ with a generalized eigenvector ${\bf T}=(1, w)^T$, for some $w \in \R^n$,
and $\tilde E_+ = E_+ \oplus$span$\{{\bf T}\}$ is the generalized eigenspace of $D\tilde F(\tilde x_0)$ corresponding to all eigenvalues whose real parts are greater than $0$. From the standard unstable manifold theorem (c.f for example \cite{CLW, chowlu, CLL}), there exists a $C^k$ manifold $W^u \subset \R^{n+1}$ in a neighborhood of $\tilde x_0$ so that (a) \, $\tilde x_0 \in W^u$; \, (b) \, $T_{\tilde x_0} W^u = \tilde E_+$; \, (c) \, it is locally invariant under the flow of \eqref{aut}; and (d)\, $\tilde x \in W^u$ if and only the solution $\tilde x(\tau)$ of \eqref{aut} with $\tilde x(0) =\tilde x$ stay close to $\tilde x_0$ for all $\tau \in (-\infty, 0]$ and $\tilde x(\tau) \to \tilde x_0$ exponentially as $\tau \to -\infty$. Changing the variable back from $\tau$ to $t$, the exponential rate $O(e^{r\tau})$ as $\tau \to -\infty$ becomes $O(t^r)$. Therefore, for
equation \eqref{sODEs}, solutions on $W^u \backslash \{t=0\}$ exactly correspond to solutions which are H\"older in $t$ and converge to $x_0$ as $t \to 0$. This proves statement (2). Statement (3) is obvious when one notices that $t = c e^\tau$. Statement (4) is a consequence of the tangency of $W^u$ to $\tilde E_+$, which is equal to $\R {\bf T}
=\R (1, 0)^T$ in this case, at $\tilde x_0$. Statement (5) follows directly from the the standard result that unstable manifolds are smooth in the external parameters \cite{CLL}.
\end{proof}

%\medskip

%We briefly illustrate how theorem applies to \eqref{standing}(F). Let $x_1=v-v_0$ and $x_2 = av_a$ and then we have $F(a, x, v_0) = (x_2, -\mu a (x_1 + v_0) - a (x_1 + v_0)^3)^T$ with the parameter $v_0$. It is clear that $F(0, 0, v_0) =0$ for all $v_0$. Moreover $F_x (0, 0, v_0)= \begin{pmatrix} 0 &1 \\ 0 &0\end{pmatrix}$ and thus $E_+ =\{0\}$. Therefore, for any $v_0$, there exists a unique solution $v(a)$ so that $v$ and $a v_a$ are both H\"older continuous in $a$ and $v(0)= v_0$. Moreover, the above statement (4) implies $x_2 = a v_a = O(a^2)$ which implies $v_a(0)=0$.

\end{document}